\documentclass[pdftex,pre,aps,nofootinbib,amssymb,amsmath,nobibnotes,twocolumn,groupedaddress,floatfix,showpacs]{revtex4-1}
\pdfoutput=1

\usepackage{graphicx}

\usepackage{enumerate}
\usepackage{amsthm}
\usepackage{xcolor}
\usepackage{psfrag}
\usepackage{hyperref}
\hypersetup{pdftitle={Narrow scope for resolution-limit-free community detection},pdfauthor={V.A. Traag, P. Van Dooren, Y. Nesterov, 2010}}

\newtheorem{definition}{Definition}
\newtheorem{theorem}{Theorem}

\newcommand{\C}{\ensuremath \mathcal{C}}
\newcommand{\D}{\ensuremath \mathcal{D}}
\newcommand{\Hf}{\ensuremath \mathcal{H}}

\definecolor{color2}{rgb}{0.894118,0.101961,0.109804}
\definecolor{color1}{rgb}{0.215686,0.494118,0.721569}
\definecolor{color3}{rgb}{0.301961,0.686275,0.290196}
\definecolor{color4}{rgb}{0.596078,0.305882,0.639216}
\definecolor{color5}{rgb}{1.000000,0.498039,0.000000}
\definecolor{color6}{rgb}{1.000000,1.000000,0.200000}

\begin{document}

\title{Narrow scope for resolution-limit-free community detection}
\author{V.A. Traag}
\email[Corresponding Author:]{vincent.traag@uclouvain.be}
\author{P. Van Dooren}
\affiliation{ICTEAM, Universit\'e Catholique de
Louvain, B\^atiment Euler, \\
Avenue G. Lema\^itre 4, B-1348 Louvain-la-Neuve, Belgium}
\author{Y. Nesterov}
\affiliation{CORE, Universit\'e Catholique de Louvain,
Voie du Roman Pays 34, B-1348 Louvain-la-Neuve, Belgium}

\date{\today}

\begin{abstract}
Detecting communities in large networks has drawn much attention over the years. While modularity remains one of the more popular methods of community detection, the so-called resolution limit remains a significant drawback. To overcome this issue, it was recently suggested that instead of comparing the network to a random null model, as is done in modularity, it should be compared to a constant factor. However, it is unclear what is meant exactly by ``resolution-limit-free'', that is, not suffering from the resolution limit. Furthermore, the question remains what other methods could be classified as resolution-limit-free. In this paper we suggest a rigorous definition and derive some basic properties of resolution-limit-free methods. More importantly, we are able to prove exactly which class of community detection methods are resolution-limit-free. Furthermore, we analyze which methods are not resolution-limit-free, suggesting there is only a limited scope for resolution-limit-free community detection methods. Finally, we provide such a natural formulation, and show it performs superbly.
\end{abstract}

\maketitle

\section{Introduction}

The last decade has seen an incredible rise in network studies, and will likely continue to rise~\cite{Lazer:2009p10015,Watts:2007p10016}. Besides the study of properties, such as degree distributions, clustering coefficients and average path length~\cite{Newman:2003p432}, many complex networks exhibit some modular structure~\cite{Fortunato:2010p6733,Porter:2009p329}. These communities might represent different functions, or sociological communities, and have been successfully studied on a wide variety of networks, ranging from metabolic networks~\cite{Guimera2005Functional} to mobile phone networks~\cite{Blondel2008Fast} and airline transportation networks~\cite{Guimera2005Worldwide}.

One of the most popular methods for community detection is that of
modularity~\cite{Newman2004Finding}. The past few years suggestions have been
made to extend or alter the original definition, for example, allowing detection
in bipartite networks~\cite{Barber:2007p7162}, networks with negative
links~\cite{Traag:2009p147}, and dynamical networks~\cite{Mucha:2010p321}.
Although modularity optimization seems to be able to accurately identify known
community structures~\cite{Lancichinetti:2008p7072}, it suffers from an inherent
problem, namely a resolution limit~\cite{Fortunato:2007p183}, which affects the
effectiveness of community detection~\cite{Good:2010p10071}. This resolution
limit prevents detection of smaller communities in large networks, although this
effect can be mitigated somewhat by a so-called resolution
parameter~\cite{kumpala07}, which can be related to time scales of random walks
on the network~\cite{Delvenne:2010p11939}. Another approach adds self-loops in
order to circumvent this resolution limit problem~\cite{Arenas:2008p12974}, and
can similarly be related to the random walk approach~\cite{Delvenne:2010p11939}.
The use of such resolution parameters enables the investigation of community
structures at various levels of description. The analysis of which levels of description are meaningful or relevant then becomes important, but we will not investigate this issue here.

Recently, a new method has been suggested that would not suffer from this
resolution limit~\cite{Ronhovde:2010p7486}. For showing a method suffers from a
resolution limit a few clear cases suffice, but the opposite seems more
difficult to argue. That is, although there is no problem for the cases analyzed,
perhaps more complex cases will show some issues not yet considered. Hence, a
proper definition of \emph{resolution-limit-free} is called for, which we will
develop in this paper. Furthermore, the question then is what methods will
suffer from this resolution limit and which not.

We will analyze this question within the framework of the first principle Potts
model as developed by Reichardt and Bornholdt~\cite{Reichardt2007Partitioning}.
Various methods can be derived from this first principle Potts model, among them
modularity, and we will briefly examine them. We will suggest a very simple
alternative, which we term the constant Potts model (CPM). It can be easily
shown that the CPM is resolution-limit-free according to our definition, but it
will follow immediately from the more general theorem we will prove. Arguably,
the CPM is the simplest formulation of any (non-trivial) resolution-limit-free
method, and can be well interpreted.

In the next section, we will briefly examine this first principle Potts model,
review some models that can be derived from it, and introduce the CPM. We will
then briefly explain the problem of the resolution limit when using modularity,
followed by the introduction of the definition of a resolution-limit-free method
(i.e. not suffering from a resolution limit), and we will show some general
properties of resolution-limit-free methods. We will then prove which methods
are resolution-limit-free and analyze which are not. Finally, we show the CPM
method performs superbly.

\section{Potts Model for Community Detection}

First, let us introduce the notation. We consider a connected graph $G=(V,E)$
with $n = |V|$ nodes and $m = |E|$ edges. The adjacency matrix $A_{ij} = 1$ if
there is an $(ij)$ edge, and $0$ otherwise. For weighted graphs the weight of a
link is denoted by $w_{ij}$, while for an unweighted graph we can consider
$w_{ij} = 1$. We denote the community of a node $i$ by $\sigma_i$.

In principle, links within communities should be relatively frequent, while
those between communities should be relatively rare. Building on this idea, we
will (i) reward links within communities; and (ii) penalize missing links within
communities~\cite{Reichardt2007Partitioning}. In general, this can then be
written as
\begin{equation}
 \Hf = -\sum_{ij} (a_{ij} A_{ij} - b_{ij} (1 - A_{ij}) )\delta(\sigma_i,\sigma_j),
  \label{equ:original}
\end{equation}
where $\delta(\sigma_i,\sigma_j) = 1$ if $\sigma_i=\sigma_j$ and zero
otherwise, and with some weights $a_{ij},b_{ij} \geq 0$. Minimal $\Hf$
correspond to desirable partitions, although such a minimum is not necessarily
unique. The choice of the weights $a_{ij}$ and $b_{ij}$ are important, and have
a definite impact on what type of communities are detected. 

\subsection{Previous methods}

In the current literature, at least four different choices exist (and presumably
some other methods may be rewritten as such), leading to four different methods
for detecting communities. We will briefly explicate these four different
approaches.

Reichardt and Bornholdt (RB) set $a_{ij} = w_{ij} - b_{ij}$ and $b_{ij} = \gamma_{RB} p_{ij}$ with a new variable $p_{ij}$ that represents the probability of a link between $i$ and $j$, known as the random null model. Working out their choice of parameters, we arrive at
\begin{equation}
 \Hf_{RB} = -\sum_{ij} (A_{ij}w_{ij} - \gamma_{RB}p_{ij})\delta(\sigma_i,\sigma_j).
\end{equation}
One of the most used null models is the so-called configuration model, which is
$p_{ij} = k_ik_j/2m$, where $k_i=\sum_j A_{ji}$ is the degree of node $i$. By
using this null model, and setting $\gamma_{RB} = 1$ we recover the original
definition of modularity~\cite{Newman2004Finding}. Independent of the exact
choice of the null model $p_{ij}$, it can be shown the method will suffer from a
resolution limit~\cite{kumpala07}, which thereby also holds for
modularity~\cite{Fortunato:2007p183}.

Another approach by Arenas, Fern\'andes and G\'omez (AFG) uses self-loops in
order to try to circumvent the resolution limit~\cite{Arenas:2008p12974}. They
do not explicitly derive their model based on the first principle Potts model,
but it can easily be done. If we set $a_{ij} = w_{ij} - b_{ij}$ and $b_{ij} =
p_{ij}(r) - r\delta(i,j)$, with $p_{ij}(r) = \frac{(k_i + r)(k_i + r)}{2m + nr}$
and $\delta(i,j) = 1$ only if $i=j$ and zero otherwise, we arrive at their model
(up to a multiplicative scaling)
\begin{equation}
  \Hf_{AFG} = -\sum_{ij} \left(A_{ij}w_{ij} + r\delta_{ij} - p_{ij}(r)\right)\delta(\sigma_i, \sigma_j)
\end{equation}
The null model $p_{ij}(r)$ is here defined as the configuration model on the
graph where a self-loop with weight $r$ is added to each node.

Ronhovde and Nussinov (RN) do not include any random null model, in order to avoid
issues with the resolution limit, and in general set $a_{ij} = w_{ij}$ and $b_{ij} =
\gamma_{RN}$ (although for specific networks, such
as with negative weights, they allow some minor changes). Working this out we
obtain
\begin{equation}
 \Hf_{RN} = -\sum_{ij} (A_{ij} (w_{ij} + \gamma_{RN}) - \gamma_{RN})\delta(\sigma_i,\sigma_j).
\end{equation}

Finally, the label propagation method~\cite{Raghavan:2007p11294} can be shown to
be equivalent to the Potts model $-\sum_{ij} A_{ij}w_{ij}
\delta(\sigma_i,\sigma_j)$~\cite{Tibely:2008p10903}, which corresponds to the
weights $a_{ij} = w_{ij}$ and $b_{ij} = 0$. This is the least interesting
formulation, since there is only one global optimum, namely all nodes in a
single community, which is trivial. However, the local minima could be of some
interest.

It is not surprising then that these four different formulations share certain
characteristics for some choice of parameters. The RB model is equivalent to the
RN model up to a multiplicative constant by using an Erd\"os-Reny\`i (ER) null
model, i.e. $p_{ij} = p$ and by setting $\gamma_{RN} = \gamma_{RB}p/(1 -
\gamma_{RB}p)$. For $\gamma_{RN} = 0$ the RN model obviously reduces to the
label propagation method. Finally, for the AFG model, when using $r=0$ we
retrieve the modularity (i.e. the RB model with configuration null model and
$\gamma_{RB} = 1$).

\subsection{Constant Potts model}
We introduce an alternative method, that uses slightly different weights. By defining $a_{ij} = w_{ij} - b_{ij}$ and $b_{ij} = \gamma$, we obtain a version that is similar to both the RB and the RN model, but is simpler and more intuitive to work with. If we work this out, we obtain the rather simple expression
\begin{equation}
  \Hf = -\sum_{ij} (A_{ij} w_{ij} - \gamma)\delta(\sigma_i,\sigma_j).
  \label{equ:our_model}
\end{equation}
Let us call this the constant Potts model (CPM), with the ``constant'' here referring to the comparison of $A_{ij}$ to the constant term $\gamma$. It is clear that this is equivalent to the RN model for unweighted graphs by setting $\gamma=\frac{\gamma_{RN}}{1 + \gamma_{RN}}$ and ignoring the multiplicative constant. Furthermore, it is equal to the RB model when setting $\gamma=\gamma_{RB}p$ for the ER null model. By setting $\gamma=0$ we retrieve the label propagation method. Also, it is highly similar to an earlier Potts model suggested by Reichardt and Bornholdt~\cite{Reichardt:2004p11243}.

If we denote the number of edges\footnote{Or technically, twice the number of
edges in an undirected graph, or the total weight in a weighted graph.} inside
community $c$ by $e_c = \sum_{ij} A_{ij}w_{ij}
\delta(\sigma_i,c)\delta(\sigma_j,c)$, and the number of nodes in community $c$
by $n_c = \sum_{i} \delta(\sigma_i,c)$, we can rewrite Eq.~\eqref{equ:our_model} as
\begin{equation}
  \Hf = -\sum_{c} e_{c} - \gamma n^2_c.
  \label{equ:our_model_comms}
\end{equation}
In other words, the model tries to maximize the number of internal edges while at the same time keeping relatively small communities. The parameter $\gamma$ balances these two imperatives. In fact, the parameter $\gamma$ acts as the inner and outer edge density threshold. That is, suppose there is a community $c$ with $e_{c}$ edges and $n_c$ nodes. Then it is better to split it into two communities $r$ and $s$ whenever
\begin{equation*}
 \frac{e_{r \leftrightarrow s}}{2n_rn_s} < \gamma,
\end{equation*}
where $e_{r \leftrightarrow s}$ is the number of links between community $r$ and $s$. This ratio is exactly the density of links between community $r$ and $s$. So, the link density between communities should be lower than $\gamma$, while the link density within communities should be higher than $\gamma$. This thus provides a clear interpretation of the $\gamma$ parameter.

In general, where $\gamma = \min_{ij} A_{ij}w_{ij}$ the optimal solution is the trivial solution of all nodes in one big community. On the other extreme, when $\gamma = \max_{ij} A_{ij} w_{ij}$, it is optimal to split all nodes in communities, i.e. such that each node forms a community by itself. In fact, communities of one node only exist when $\gamma = \max_{ij} A_{ij} w_{ij}$, since otherwise it will always be beneficial to put the node in one of its neighbors' communities. Hence, for practical purposes $\min_{ij} A_{ij} w_{ij} \leq \gamma \leq \max_{ij} A_{ij} w_{ij}$.

\begin{figure*}
  \includegraphics{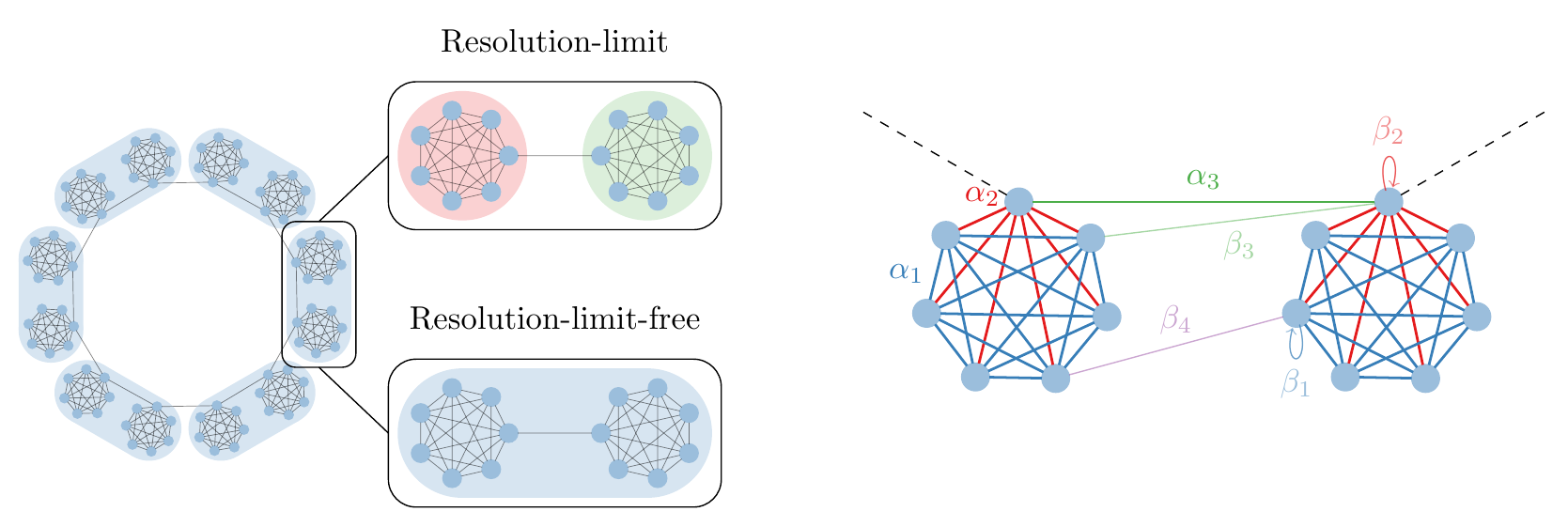}
  \caption{(Color online) The problem of the resolution limit with modularity is
  usually demonstrated on a ring network of cliques. The cliques are as densely
  connected as possible, and as sparsely connected between them, while still
  retaining a connected graph. The resolution limit is said to arise because it
  will merge cliques depending on the size of the network. In fact, methods that
  do not suffer from the resolution limit, i.e. resolution-limit-free methods,
  may merge these cliques also. However, the distinguishing fact between
  resolution-limit and resolution-limit-free methods is that the first will
  detect smaller subcommunities when applied to the subgraph (i.e. not the same
  subpartition), while the latter will not detect smaller subcommunities (i.e.
  the same subpartition will remain optimal on the subgraph). Of course, whether
  the communities should consist of only cliques or of multiple joined cliques
  will still depend on the actual resolution of the method. For CPM this
  resolution parameter is designated by $\gamma$. So, for a particular value of
  $\gamma$ the cliques will be separated, while for another value they will be
  merged. When analyzing the property of resolution-limit-free more in detail,
  we will investigate this ring of clique network more closely. The actual
  weights $a_{ij}$ and $b_{ij}$ need to be the same on isomorphic graphs, which
  restricts the possible number of different weights. We denote these different
  possibilities with $\alpha$'s for the present links (thicker lines) and
  $\beta$'s for the missing links (thinner lines). Please refer to the main text
  for further details.}
  \label{fig:example}
\end{figure*}

\section{Resolution limit}

Traditionally the resolution limit is investigated by analyzing the
counterintuitive merging of communities~\cite{Fortunato:2007p183}, for example
cliques or some smaller communities that are only sparsely interconnected as
displayed in Fig.~\ref{fig:example}. The RB model with a configuration null
model for example will merge two neighboring cliques in this ring network of
cliques when~\cite{kumpala07} $\gamma_{RB} < q/(n_c(n_c - 1) + 2)$, where $q$ is
the number of cliques and $n_c$ is the number of nodes of a clique. Since the
number of cliques $q$ is a global variable, it shows modularity might be ``hiding''
some smaller communities within larger communities, depending on the size of the
network. Indeed in~\cite{Fortunato:2007p183} it was suggested to look at each
community to consider whether it had any sub communities. Some different, though
related, problems with modularity were noticed in~\cite{Brandes:2007p196} and
more recently in~\cite{Krings:2011p12227}.

The AFG model considers self-loops of a certain weight to overcome this
problem~\cite{Arenas:2008p12974}. Yet the model still depends on a null
model, and so it is not surprising to find that the merging still depends on
some global parameters. The implicit inequality for merging two cliques in the
ring network of cliques is $q > n_c(n_c - 1) + 2 + n_cr$, which for $\gamma_{RB}
= 1$ and $r = 0$ matches the previous result. Although the resolution parameter
$r$ might be used to investigate the community structure at various scales
similar to the $\gamma_{RB}$ resolution parameter, it does not fundamentally
address the issues of the resolution limit.

The RN model, on the other hand, will only join two cliques
when~\cite{Ronhovde:2010p7486} $\gamma_{RN} < 1/(n^2_c - 1)$,
which does not depend on the number of cliques $q$, and depends only on the
local variable $n_c$, so is argued not to suffer from any resolution limit.
For the CPM suggested here, we arrive at the condition $\gamma < 1/n^2_c$,
which also does not depend on the number of cliques $q$ and can hence also said
to be resolution-limit-free. More general, CPM favors to cluster $r$ consecutive
cliques instead of $r-1$ at the point when $\gamma < 2/(r(r-1)n_c^2)$.

However, it remains somewhat unclear what is meant exactly by
resolution-limit-free in the above discussion, and the label
\emph{resolution-limit-free} requires a more precise definition. Consider for
example that we take away the dependence on the number of links in the
configuration null-model, so that we take $p_{ij} = k_ik_j$. Notice that this
only corresponds to a multiplicative rescaling of $\gamma_{RB}$ by $2m$.
Revisiting the case above, we come to the conclusion that cliques are separated
whenever $2\gamma_{RB} > (n_c(n_c - 1) + 2)^{-2}$, which unsurprisingly no
longer depends on any global variables. By the argument employed previously, the
method should be resolution-limit-free.

Not all problems have disappeared however. Suppose we take the subgraph
consisting of only two of these cliques. We analyze when the method would merge
the two cliques in this subgraph, which is the case whenever $2\gamma_{RB} <
(n_c(n_c - 1) + 1)^{-2}$. Even though neither inequality depends on any global
variables, a problem remains. Combining the above two inequalities, we obtain
that whenever
\begin{equation*}
	(n_c(n_c - 1) + 2)^{-2} <  2\gamma_{RB} < (n_c(n_c - 1) + 1)^{-2},
\end{equation*}
the method will separate the cliques in the larger graph, yet merge them in the subgraph.

The above discussion motivates us to consider the following definition of a
\emph{resolution-limit-free} method. The general idea is that when looking at
any induced subgraph of the original graph, the partitioning results should not
be changed. In order to introduce this definition, let $\Hf$ be any objective
function (which we want to minimize), we then call a partition $\C$ for a graph
$\Hf$-optimal whenever $\Hf(\C) \leq \Hf(\C')$ for any other partition $\C'$. We
can then define \emph{resolution-limit-free} as follows.

\begin{definition}
Let $\C = \{C_1, C_2, \ldots, C_q\}$ be a $\Hf$-optimal partition of a graph $G$. Then the objective function $\Hf$ is called \emph{resolution-limit-free} if for each subgraph $H$ induced by $\D \subset \C$, the partition $\D$ is also $\Hf$-optimal.
\end{definition}

Furthermore, we introduce the notion of additive objective functions.

\begin{definition}
An objective function $\Hf$ for a partition $\C = \{C_1,\ldots, C_q\}$ is called additive whenever $\Hf(\C) = \sum_i \Hf(C_i)$, where $\Hf(C_i)$ is the objective function defined on the subgraph $H$ induced by $C_i$.
\end{definition}

If we have an $\Hf$-optimal partition $\C$ for an additive resolution-limit-free objective function $\Hf$, we can replace subpartitions of $\C$ by other optimal subpartitions.

\begin{theorem}
Given an additive resolution-limit-free objective function $\Hf$, let $\C$ be an $\Hf$-optimal partition of a graph $G$ and let $H \subset G$ be the induced subgraph by $\D \subset \C$. If $\D'$ is an alternative optimal partition of $H$ then $\C' = \C \setminus \D \cup \D'$ is also $\Hf$-optimal.
\end{theorem}
\begin{proof}
Define $\C'$ and $\D'$ as in the theorem. By additivity, $\Hf(\C') = \Hf(\C \setminus \D) + \Hf(\D')$, and by optimality $\Hf(\D') \leq H(\D)$. Since also $\Hf(\C) = \Hf(\C \setminus \D) + \Hf(\D)$ we obtain $\Hf(\C') \leq \Hf(\C)$, so $\C'$ is also optimal.
\end{proof}

Although, this might seem to contradict the NP-hardness of community detection
methods, this is not the case. It states that when there are two optimal
partitions, any combination of those partitions are optimal, so in a certain
sense, they are spanning a space of optimal partitions. It does not say whether
such a partition can be easily found. Also, there might be two optimal
partitions that cannot be obtained by recombining them, because all communities
partly overlap with each other.

It is also possible to prove that a complete graph $K_n$ with $n$ nodes is never
split (unless into all nodes separately).

\begin{theorem}
Given a resolution-limit-free objective function $\Hf$, the $\Hf$-optimal partition of $K_n$ for all $n$ is either only one community, namely all nodes, or $n$ communities consisting each of one node.
\end{theorem}
\begin{proof}
Assume on the contrary there is an optimal partition $\C$ of $K_n$ such that $1 < |\C| < n$. Then for any $\D \subset \C$ the subgraph $H$ induced by $\D$ is a complete graph. But by assumption, $\D$ is then not optimal, and by resolution-limit-free, $\C$ is then not optimal. Hence, inductively, the theorem must hold for all $n$.
\end{proof}

Also notice that a resolution-limit-free method will never depend on the size of the
network to merge cliques in the ring of cliques network. This can be easily seen
from the fact that a subgraph of a large ring of cliques network also appears in
a smaller ring. So if the method would merge cliques in some large graph, by
the resolution-limit-free property, it would also need to merge them together in
the smaller graph. Hence, the actual merging cannot depend on the size of the
network. In this sense it captures this prior concept of the resolution limit.

Equipped with this definition, we can analyze the first principle Potts model
further. For example, what conditions should be imposed on the weights $a_{ij}$
and $b_{ij}$ in Eq.~\eqref{equ:original} for the method to be
resolution-limit-free? Would a method that takes into account the local number
of triangles be resolution-limit-free? Or would it be possible to use the
shortest (weighted) path for example?

We can prove that CPM is \emph{resolution-limit-free} in this sense, just like
the RN model and the LP model. The CPM model is also trivially shown to be
additive by Eq.~\eqref{equ:our_model_comms}. Perhaps it is less obvious, but the
RB model is not additive, since it cannot be defined in terms of
\emph{independent} contributions, i.e. the contribution $\Hf(C_i)$ per community
depends on the whole graph $G$, instead of only on the subgraph $H$ induced by
$C_i$. Nor is the RB model resolution-limit-free according to our definition,
regardless of the null model~\cite{kumpala07}, and hence modularity is not
resolution-limit-free.  Furthermore, as we have seen, also when using $p_{ij} =
k_ik_j$ the model is not resolution-limit-free. Finally, the AFG model is not
resolution-limit-free either. 

Since the CPM model is also related to the RB model using the ER null-model, it
is tempting to conclude it is also resolution-limit-free. Indeed, this might be
said to be the case, if we choose $p$ independently of the graph, i.e. not
define it as $p = m/n(n-1)$, and simply choose it as some value $p \in
\mathbb{R}$. However, we then obviously retrieve the CPM model. This shows that
resolution-limit-free methods are strongly constrained, and there is only a fine
line between resolution-limit and resolution-limit-free methods.

This follows from the more general theorem we will now prove. For this, we first
introduce the notion of local weights. Again, building on the idea of
subgraphs, we define local weights as weights that do not change when looking to
subgraphs.

\begin{definition}
  Let $G$ be a graph, and let $a_{ij}$ and $b_{ij}$ as in Eq.~\eqref{equ:original} be the associated weights. Let $H$ be a subgraph of $G$ with associated weights $a'_{ij}$ and $b'_{ij}$. Then the weights are called \emph{local} if $a_{ij} = \lambda a'_{ij}$ and $b_{ij} = \lambda b'_{ij}$, where $\lambda = \lambda(H) > 0$ can depend on the subgraph $H$.
\end{definition}
Clearly then, the RN and CPM model have local weights, while the RB and AFG
model do not.  This definition says that local weights should be independent of
the graph $G$ in a certain sense. In fact, it is quite a strong requirement, as
it should even hold for a single link $(ij)$ in the subgraph where only $i$ and
$j$ are included. That means it can not depend on any other link \emph{but} the
very link itself. Since for missing links, there is (usually) no associated
weight or anything, it can only be constant. There are some exceptions, such as
multipartite networks, or networks embedded in geographical
space~\cite{Expert:2011p12973,Lambiotte:2008p2987}, where some sensible
non-constant local weights can be provided. Hence, the RN model and the CPM
model are one of the few sensible options available for having local variables.
We can now prove the more general statement that methods using local weights are
resolution-limit-free.

\begin{theorem}
 The objective function $\Hf$ as defined in Eq.~\eqref{equ:original} is \emph{resolution-limit-free} if it has \emph{local} weights.
\end{theorem}

\begin{proof}
  Let $\C$ be the optimal partition for $G$ with community assignments $c_i$, $\D \subset \C$ a subset of this partition, and $H$ the subgraph induced by $\D$ with $h$ nodes. Furthermore, we denote by $d_i$ the community indices of $\D$, such that $d_i = c_i$ for $1 \leq i \leq h$ and by $A'$ the adjacency matrix of $H$, so that $A_{ij} = A'_{ij}$ for $1 \leq i \leq h$. Assume $\D$ is not optimal for $H$, and that $\D^*$ is optimal, such that $\Hf(\D) > \Hf(\D^*)$. Then define $c^*$ by setting $c^*_i = d^*_i$ for $1 \leq i \leq h$ and $c^*_i = c_i$ for $h < i \leq n$. Then because the result is unchanged for the nodes $h < i \leq n$, we have that
  \begin{equation*}
    \Delta \Hf = \Hf(\C) - \Hf(\C^*) = \frac{1}{\lambda}(\Hf(\D) - \Hf(\D^*))  > 0
  \end{equation*}
  where the last step follows from the locality of the weights $a_{ij}$ and $b_{ij}$. This inequality contradicts the optimality of $\C$. Hence, for all induced subgraphs $H$, the partition $\D$ is optimal, and the objective function $\Hf$ is \emph{resolution-limit-free}.
\end{proof}

The converse is unfortunately not true. Consider a graph $G$ with some weights $a_{ij}$ and $b_{ij}$. Then pick a subgraph $H$ induced by some subpartition $\D$, and define the weights $a'_{ij} = a_{ij}$ and $b'_{ij} = b_{ij}$ except for one particular edge $(kl)$, for which we set $a'_{kl} = a_{kl} + \epsilon$. Then for some $\epsilon > 0$, the original subpartition will remain optimal in $H$, while the weights are not local. Since the small change of the weight is \emph{only} made when considering the graph $H$, all other subpartitions will always remain optimal. Of course, such a definition of the weight is rather odd, so in practice we will never use it.

Even though the converse is not true, we can say a bit more. The weights can be
a bit different indeed, but there is not that much room for these differences.
We demonstrate this on the ring network of cliques. The weights can depend only
on the graph, so if $G$ and $G'$ are two isomorphic graphs, then $a_{ij}(G) =
a_{i'j'}(G')$, where $i$ and $i'$ are two isomorphic nodes. Hence, only a number
of weights can be different from each other in the ring network, as illustrated
in Fig.~\ref{fig:example}. All nodes within a clique are isomorphic, except the
node that connects to other cliques. So, all the edges among those $n_c - 1$
nodes are similar, and will have the same weight $\alpha_1$. All edges from
these $n_c - 1$ nodes to the ``outside'' node will have the same weight
$\alpha_2$. Finally, the edge connecting two cliques is denoted by $\alpha_3$.
The missing self-loop for the special outside node is denote by $\beta_2$
while the missing self-loop for the other nodes in the cliques is denoted by
$\beta_1$. Finally, there is (1) a missing link between the outside node and a
normal node denoted by $\beta_3$; and (2) a missing link between two normal
nodes, denoted by $\beta_4$. These weights are illustrated in Fig.~\ref{fig:example}.

Let us now analyze when the method will not be resolution-limit-free. Then, the cliques must be merged in some (large) graph, while for the subgraph consisting of these two merged cliques, they should be separated by the method. Or conversely, they should be separated in some (large) graph, but merged in the subgraph. We can write the $\Hf_s$ for all $q$ cliques being separate as
  \begin{multline*}
  \Hf_s = -q( 
	    {\color{color1}\alpha_1(n_c - 1)(n_c - 2)} 
	  + {\color{color2}2\alpha_2 (n_c - 1)} \\
	  - {\color{color1}(n_c - 1) \beta_1} 
	  - {\color{color2}\beta_2} )
  \end{multline*}
  and $\Hf_m$ for merging all two consecutive cliques as
  \begin{multline*}
  \Hf_m = -\frac{q}{2} 2(
	    {\color{color1}\alpha_1(n_c - 1)(n_c - 2)} 
	  + {\color{color2}2\alpha_2 (n_c - 1)} \\
	  - {\color{color1}(n_c - 1)\beta_1}
	  - {\color{color2}\beta_2} 
	  + {\color{color3}\alpha_3 } 
	  - {\color{color3}\beta_3(n_c - 1)} 
	  - {\color{color4}\beta_4(n_c - 1)^2} ) 
  \end{multline*}  
  Furthermore, for the induced subgraph $H$ consisting of two consecutive cliques, we can write $\Hf'_s$ for separating the two cliques and $\Hf'_m$ for merging them, similarly as before, where $\alpha'$ and $\beta'$ are the weights for the subgraph $H$. Then the method is not resolution-limit-free if it would merge the two cliques at a higher level (i.e. when $\Hf_m < \Hf_s$) yet would not merge them at smaller scale (i.e. when $\Hf'_s < \Hf'_m$), or vice versa. Working out this condition for $\Hf_m < \Hf_s$ (and similarly for $\Hf_m > \Hf_s$) gives us
  \begin{equation*}
    {\color{color3}\alpha_3} > (n_c - 1)({\color{color4}\beta_4(n_c - 1)} + {\color{color3}\beta_3}),
  \end{equation*}
  while for $\Hf'_s < \Hf'_m$ (and similarly for $\Hf'_s > \Hf'_m$) we obtain
  \begin{equation*}
    {\color{color3}\alpha'_3} < (n_c - 1)({\color{color4}\beta'_4(n_c - 1)} + {\color{color3}\beta'_3}).
  \end{equation*}  
  Combining these two inequalities for both cases we obtain
  \begin{alignat}{3}
    {\color{color3}\alpha'_3}({\color{color4}\beta_4(n_c - 1)} + {\color{color3}\beta_3}) &< {\color{color3}\alpha_3}({\color{color4}\beta'_4(n_c - 1)} + {\color{color3}\beta'_3}),  
    \label{equ:ms} \\
    {\color{color3}\alpha'_3}({\color{color4}\beta_4(n_c - 1)} + {\color{color3}\beta_3}) &> {\color{color3}\alpha_3}({\color{color4}\beta'_4(n_c - 1)} + {\color{color3}\beta'_3}).
    \label{equ:sm}
  \end{alignat}
  where either Eq.~\eqref{equ:ms} or~\eqref{equ:sm} should hold. Hence, only if the left hand side equals the right hand side, it does not constitute a counter example. Working out this equality, there are two possibilities. Either the weights should be local, or the following equality should hold
  \begin{equation*}
  	n_c - 1 = \frac{
  						{\color{color3}\alpha_3}{\color{color4}\beta'_3}
	  					-
  						{\color{color3}\alpha'_3}{\color{color4}\beta_3}
  					}
  					{
  						{\color{color3}\alpha'_3}{\color{color4}\beta_4}
	  					-
  						{\color{color3}\alpha_3}{\color{color4}\beta'_4}
  					}.
  \end{equation*}  
Obviously, this again constitutes some very particular case of non-local weights. We can repeat this same procedure for other subpartitions, and for other graphs, thereby forcing the weights to be of a very particular kind. This thus leaves little room for having any sensible non-local definition such that the method is resolution-limit-free.

This means resolution-limit-free community detection has only a quite limited
scope. In fact, the CPM seems to be the simplest non-trivial sensible
formulation of any general resolution-limit-free method, although there is some
leeway for special graphs (i.e. having some node properties, such as
multipartite graphs). This is not to say that methods with non-local weights
(e.g. modularity, AFG, number of triangles, shortest path, betweenness) should
never be used for community detection at all, they are just never
resolution-limit-free.

\section{Performance}

In order to assess the performance of the proposed CPM model, we performed
various tests. Using the latest suggested test
networks~\cite{Lancichinetti:2008p7072} we find that the CPM model and the
accompanying algorithm is both very accurate and efficient. More
details on the efficient Louvain-like algorithm, the test procedure and the
calculations on the resolution parameters can be found in the appendix at the
end of this article.

We have examined both directed test networks as well as hierarchical test
networks, where communities exist at multiple levels in the data. Communities
become less discernible for higher values of the parameter $\mu$ of having
links outside its community. For hierarchical communities, there are two such
parameters: $\mu_1$ for the first level (the large communities), and $\mu_2$ for
the second level (the subcommunities). These mixing parameters allow us to
calculate what the inner and outer densities of communities are. We exploit this
fact to calculate the proper $\gamma=\gamma^*$ in order to investigate the performance of
the CPM, and similarly $\gamma_{RB}=\gamma^*_{RB}$ for the RB model using the configuration null model. This way, the results do not depend on any particular method to
determine the correct parameter, which represents another challenging problem.

Some of the earlier algorithms and models that showed excellent performance~\cite{Lancichinetti:2009p7112} are
the Louvain~\cite{Blondel2008Fast} method for optimizing modularity, and the
Infomap method~\cite{Rosvall2007Informationtheoretic}. In
Fig.~\ref{fig:performance} we have displayed the results for (1) the CPM model; (2) the RB model using an ER null model (i.e. CPM with $\gamma=p$); (3) the RB model using the configuration null model\footnote{Since we use directed test networks,
we use a small adjustment to use the directed configuration null
model~\cite{Leicht2008Community}} with ``corrected'' parameter value $\gamma_{RB}^*$; (4) the modularity model, (i.e. RB using the configuration null model and $\gamma_{RB}=1$); and finally (5) the Infomap method. We have performed tests on networks having $n=10^3$ and $n=10^4$ nodes, with a degree distribution exponent of $2$ (with average degree $15$ and maximum degree $50$) and community size distribution exponent $1$ (with community sizes ranging from $20$ to $100$). Per value of $\mu$ 100 graphs were used to obtain this result. 

It can be clearly seen that CPM performs extremely well. The difference in performance of the CPM model in comparison to the RB model using the ER null model is especially striking. This is not a consequence of the method being resolution-limit-free or not, but it rather depends on choosing the correct resolution parameter. Obviously then, setting $\gamma = p$ is in general
not a very good strategy, and for general networks one should carefully analyze
at which resolution the network contains meaningful partitions.

\begin{figure}[bt]
    \centering
    \includegraphics{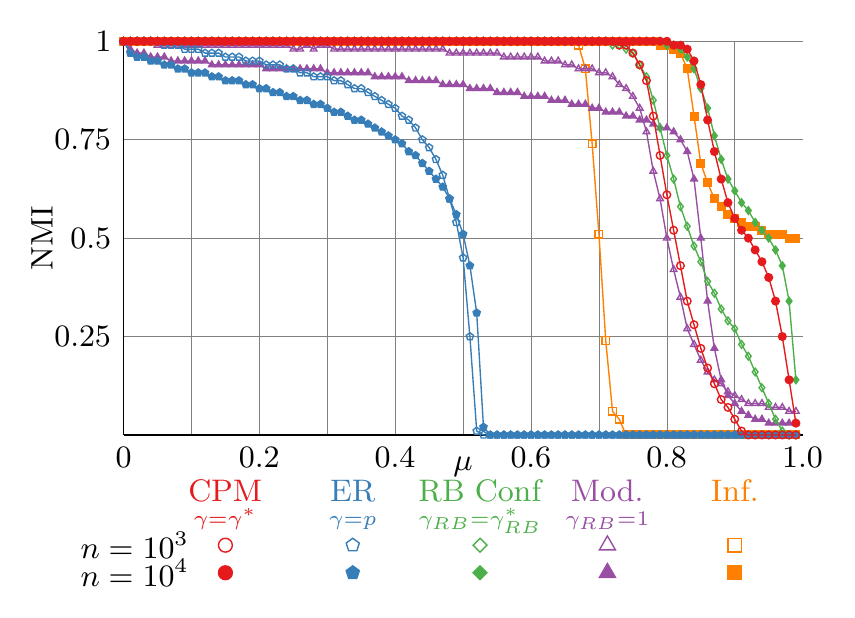}
    \caption{(Color online) Performance of various community detection methods in terms of Normalized Mutual Information (NMI) depending on mixing parameter $\mu$. Results are shown for (1) CPM, using the calculated resolution parameter $\gamma^*$; (2) the RB model with an ER random null model (i.e. corresponding to CPM using $\gamma=p$); (3) the RB model with a configuration null model with a calculated resolution parameter $\gamma_{RB}^*$; (4) modularity, in other words, the RB model with a configuration model using $\gamma_{RB}=1$; and finally (5) the Infomap method. The open symbols denote results for $n=10^3$ and the closed symbols for $n =10^4$. \label{fig:performance}}
  
\end{figure}

\begin{figure}[bt]
    \centering
   	\includegraphics{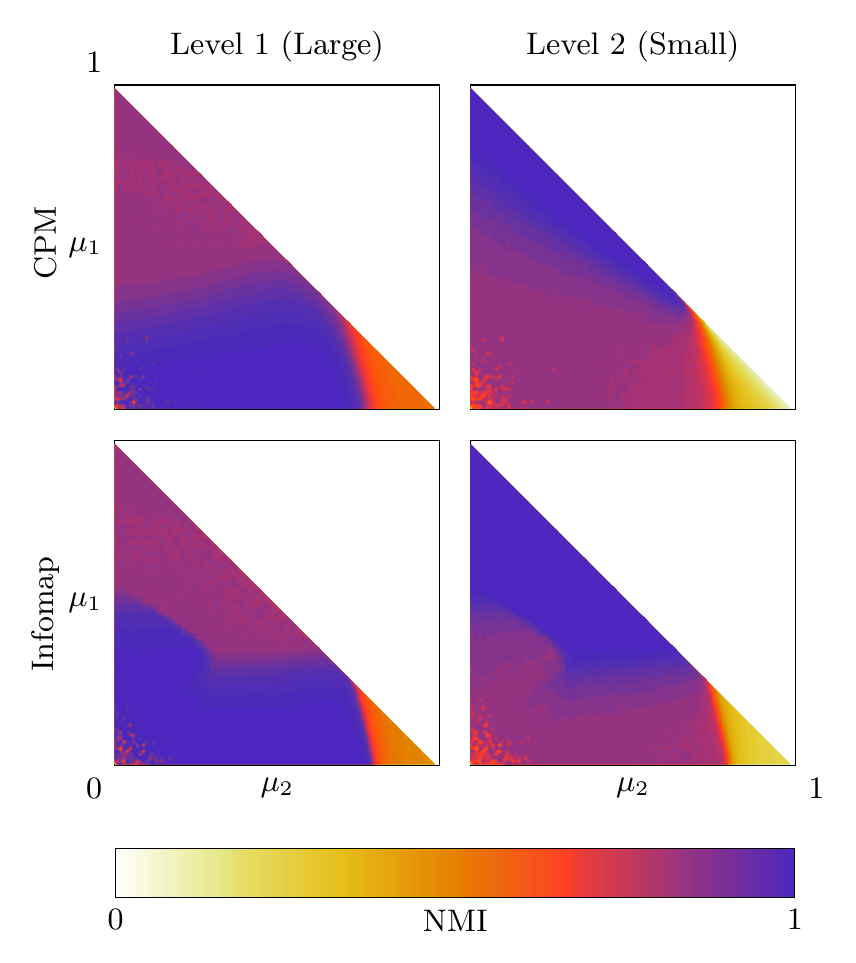}     
   	\caption{(Color online) Performance of CPM and Infomap on a hierarchical benchmark network in terms of Normalized Mutual Information (NMI) depending on mixing parameters $\mu_1$ and $\mu_2$. The networks had $n =10^4$ nodes with a degree distribution exponent of $2$ (with average degree $20$ and maximum degree $50$) and community size distribution exponent $1$ for both small (size ranging from $10$ to $50$) and large communities (size ranging from $50$ to $300$). Per combination of parameters 10 graphs were used to obtain this result. The resolution parameters $\gamma$ for the two different levels were calculated analytically for CPM.\label{fig:performance_hierarchical}}
\end{figure}

A similar effect also shows for modularity (or the RB model using the configuration model), such that when $\gamma_{RB}$ is chosen appropriately (i.e. using $\gamma_{RB} = \gamma_{RB}^*$) the method will perform better than at the ordinary resolution $\gamma_{RB}=1$. Indeed, the results of the CPM model and the RB model using the configuration null model using $\gamma^*_{RB}$ are rather comparable, although the latter's performance drops less quickly, and then outperforms CPM. Interestingly, when we use the ordinary resolution $\gamma_{RB}=1$, it becomes more difficult to detect communities in large networks using the configuration model. This constrasts with the results when we choose the appropriate resolution parameter $\gamma^*$, $\gamma_{RB}^*$ and indeed also for the Infomap method. Indeed it can be shown that the communities should become more clearly discernible for larger networks when the community sizes remain similar.

Surprisingly, both methods outperform the Infomap method, which performed superbly in previous tests~\cite{Lancichinetti:2009p7112}, when the appropriate resolution parameter is chosen. This show that determining the correct or meaningful resolution is an important issue. This remains a challenging problem, and various methods have been proposed to do so~\cite{Fortunato:2010p6733}, for example by looking at the stability of multiple (randomized) runs of an algorithm~\cite{Lambiotte:2010p6743,Ronhovde:2010p7486}, by looking for large ranges of the parameter over which the results remain stable~\cite{Arenas:2008p12974}, investigating the stability when the network is slightly perturbed~\cite{Gfeller:2005p13024} or by looking at how significant the partition is compared to a graph ensemble~\cite{Bianconi:2009p13095}. 

We have also performed extensive tests on hierarchical networks, where the
method also performs well, and is able to extract the two different levels
of communities effectively, as displayed in Fig.~\ref{fig:performance_hierarchical}. For relatively low $\mu_2 \lessapprox 0.7$, the first (larger) level becomes more clear for low $\mu_1$, while the second (smaller) level becomes more clear for larger $\mu_1$. This is both the
case for a recent hierarchical version of the Infomap method~\cite{Rosvall:2011p11938} and the CPM method. The Infomap method seems to be slightly better at detecting the correct communities, but the CPM method remains highly competitive. The possibility for having various scales of description of the network seems important, as many networks seem to have at least some hierarchical structure.

\section{Conclusion}

Several community detection methods, among which modularity, are affected by the problem of the resolution limit. In this paper we have provided a novel rigorous definition of what it means for a community detection method to be resolution (limit) free. Most importantly, we are able to prove exactly which community detection methods are resolution-limit-free, namely those methods that use local weights. This also clarifies the relationship between local methods and the resolution limit. However, we do not address the issue of determining an actual meaningful resolution, which remains a challening problem.

Moreover, there does not seem to be much room for having resolution-limit-free methods without local weights. Of the few possibilities available for having resolution-limit-free community detection, the constant Potts Model (CPM) we introduced in this paper seems to be the simplest possible formulation, and performs excellent. A rigorous definition of resolution-limit-free community detection allows for a more articulate analysis, and induces further progress on developing novel and meaningful methods.

\begin{acknowledgments}
  We acknowledge support from a grant ``Actions de recherche concert\'ees --- Large Graphs and Networks'' of the ``Communaut\'e Fran\c caise de Belgique'' and from the Belgian Network DYSCO (Dynamical Systems, Control, and Optimization), funded by the Interuniversity Attraction Poles Programme, initiated by the Belgian State, Science Policy Office. The authors would like to thank Vincent Blondel, Arnaud Browet, Jean-Charles Delvenne, Renaud Lambiotte, Gautier Krings and Julien Hendrickx for helpful comments and discussions.
\end{acknowledgments}

\appendix

\section{Louvain like Algorithm}

The algorithm we employ is derived from the Louvain
method~\cite{Blondel2008Fast}. We use the concept of node size, denoted by $n_i$
for a node $i$, initialized to $n_i = 1$ (indeed the community size $n_c =
\sum_{i} n_i \delta(\sigma_i,c)$ is related). We first iterate (randomly) over
all nodes, and put nodes greedily into the community that minimizes Eq.~\eqref{equ:our_model}. We subsequently create a new graph based on the communities, and new node sizes, and reiterate over this new smaller graph. More specifically:
\begin{enumerate}
\item Initialize $A_{ij} = w_{ij}$, with $w_{ij} = 1$ in the case of unweighted networks, and set $n_i = 1$ for all nodes $i$.
\item Loop over nodes $i$, remove it from its community and calculate for each community $c$ the increase if we would put node $i$ into community $c$,
\begin{equation}
	\Delta \Hf(\sigma_i = c) = -(e_{i \leftrightarrow c} - 2 \gamma n_i \sum_j n_j \delta(\sigma_j, c)),
\end{equation}
where $e_{i \leftrightarrow c} = \sum_{j} (A_{ij} + A_{ji})\delta(\sigma_j,c)$ is the number of edges between node $i$ and community $c$.
We put node $i$ into the community $c$ for which $\Delta \Hf(\sigma_i = c)$ is
minimal. We iterate until we can no longer decrease the objective function. \label{enum:step_move_nodes}
\item We build a new graph $A'_{cd} \allowbreak= \sum_{ij} A_{ij} \delta(\sigma_i, c)\delta(\sigma_j,d)$ and node sizes $n'_c = \sum_{i} n_i \delta(\sigma_i, c)$. We repeat step~\ref{enum:step_move_nodes} by setting $A = A'$ and $n = n'$ until the objective function can no longer be decreased.
\end{enumerate}
The implementation of the algorithm in C++ can be downloaded from the author's website: \url{http://perso.uclouvain.be/vincent.traag}.

Notice that for resolution-limit-free methods, the results should be unchanged on subgraphs. Hence, we could therefore perform the method recursively on subgraphs. We suggest then the following improvement. First we cut the network at each recursive call, until the density of the subgraph exceeds $\gamma$. Then, we recombine the subgraphs, and loop over nodes/communities to find improvements until we can no longer increase greedily, and return to the previous recursive function call. These calls should be easily parallelized, making community detection in even larger graphs or in an on-line setting possible by using cluster computing.

\section{Benchmark tests}

The benchmark networks are created using a known community structure, i.e. a
planted community structure. The community sizes $n_c$ are chosen from a
distribution following a power-law $\Pr(n_c = n) \sim n^{-\tau_2}$. The degrees
$k_i$ of the nodes are also chosen from a power-law distribution $\Pr(k_i = k)
\sim k^{-\tau_1}$.  The stubs are then connected, with probability $1 - \mu$
within a community, and with probability $\mu$ between two communities. A lower
bound $\underline{n_c}$ and upper bound $\overline{n_c}$ on the community sizes
is imposed, while for the degree the average degree $\langle k \rangle$ is
specified. For the hierarchical version, there are two levels, with the
communities of the second level embedded in the first level. A fraction of
$\mu_1$ of the links is placed between two different macro communities at the
first level, while a fraction of $\mu_2$ of the links are placed between the
small communities of the second level (but within the same large community).

Instead of detecting the resolution algorithmically, we calculate the proper resolution parameter value $\gamma$ analytically (and therefore, beforehand). In order to do so, we consider the following. The resolution parameter $\gamma$ acts as a sort of threshold on inner and outer community density. If we were to set $\gamma$ equal to the inner density, it would be rather difficult to fulfill the condition that the inner density should be higher than that, and similarly so for $\gamma$ equal to the outer density. So, we need to be as far as possible from both the inner density as well as the outer density, which would be simply the average of the two.

The inner density for a community having $n_c$ nodes can be easily found as
\begin{equation}
  p_{in} = \frac{(1 - \mu)\langle k \rangle}{n_c - 1},
\end{equation}
and the outer density (i.e. all the edges originating from a community to the
outside) is 
\begin{equation}
p_{out} = \frac{\mu \langle k \rangle}{n - n_c},
\end{equation}
where $n$ is simply the total number of nodes. The average community size
$\langle n_c \rangle$, which is proportional to
\begin{equation}
  \langle n_c \rangle \sim \sum_{n=\underline{n_c}}^{\overline{n_c}} n n^{-\beta},
\end{equation}
where $\underline{n_c}$ is the minimal community size and $\overline{n_c}$ the
maximal community size, than gives us the $\langle p_{in} \rangle$ and $\langle
p_{out} \rangle$ for the average community size. The best resolution parameter is then $\gamma^* = \frac{1}{2} (\langle p_{in} \rangle + \langle p_{out} \rangle)$.

For the hierarchical test networks, we can perform a similar analysis, and use the average of the inner and outer density, similar as before, for the two different levels. Ordinarily, the communities are assumed to exist whenever $p_{in } > p_{out}$.

For modularity, we can also calculate similar bounds. When we define by $e_c$ the number of edges within community $c$ and by $[e_c]$ the number of expected edges within a community, modularity can be written as 
\begin{equation}
  \Hf = -\sum_c e_c - \gamma_{RB} [e_c]. 
\end{equation}
Hence, each community should have a ``expected density'' or ``degree density'' $\tilde{p}_{in} = e_c/[e_c]$ within communities lower than $\gamma$, while the outer degree density should be lower between communities. Writing this out in terms of the configuration model, given the model of the test networks, we arrive at
\begin{align}
 \tilde{p}_{in}  &= \frac{(1 - \mu)n}{n_c}, \\
 \tilde{p}_{out} &= \frac{\mu n}{n - n_c}.
\end{align}
These degree densities lack a clear interpretation, in contrast with CPM. Similar as before we simply set $\gamma^*_{RB} = \frac{1}{2} (\langle \tilde{p}_{in} \rangle + \langle \tilde{p}_{out} \rangle)$ for the average community size $\langle n_c \rangle$.

For comparing our results to the known community structure, we use the normalized mutual information. Given two different partitions $C$ and $D$, the mutual information $I$ is defined as
\begin{equation*}
 I(C,D) = - \sum_{r,s} \frac{n_{r,s}}{n} \log n \frac{n_{r,s}}{n_r n_s}
\end{equation*}
with $n_{r,s}$ being the number of nodes that are in community $r$ in partition $C$ \emph{and} in community $s$ in partition $D$, while $n_r$ simply denotes the number of nodes in community $r$. The normalized mutual information $I_n(C,D)$ is then defined as
\begin{equation*}
 I_n(C,D) = \frac{2I(C,D)}{H(C) + H(D)},
\end{equation*}
where $H(C)$ indicates the entropy of a partition $C$, which is defined as
\begin{equation*}
 H(C) = - \sum_{c} \frac{n_c}{n} \log \frac{n_c}{n}.
\end{equation*}
The normalized mutual information $0 \leq I_n(C, D) \leq 1$, with $1$ indicating equivalent partitions.


\begin{thebibliography}{10}%
\makeatletter
\providecommand \@ifxundefined [1]{%
 \ifx #1\undefined \expandafter \@firstoftwo
 \else \expandafter \@secondoftwo
\fi
}%
\providecommand \@ifnum [1]{%
 \ifnum #1\expandafter \@firstoftwo
 \else \expandafter \@secondoftwo
\fi
}%
\providecommand \enquote [1]{``#1''}%
\providecommand \bibnamefont  [1]{#1}%
\providecommand \bibfnamefont [1]{#1}%
\providecommand \citenamefont [1]{#1}%
\providecommand\href[0]{\@sanitize\@href}%
\providecommand\@href[1]{\endgroup\@@startlink{#1}\endgroup\@@href}%
\providecommand\@@href[1]{#1\@@endlink}%
\providecommand \@sanitize [0]{\begingroup\catcode`\&12\catcode`\#12\relax}%
\@ifxundefined \pdfoutput {\@firstoftwo}{%
 \@ifnum{\z@=\pdfoutput}{\@firstoftwo}{\@secondoftwo}%
}{%
 \providecommand\@@startlink[1]{\leavevmode}%
 \providecommand\@@endlink[0]{}%
}{%
 \providecommand\@@startlink[1]{%
  \leavevmode
  \pdfstartlink
   attr{/Border[0 0 1 ]/H/I/C[0 1 1]}%
   user{/Subtype/Link/A<</Type/Action/S/URI/URI(#1)>>}%
  \relax
 }%
 \providecommand\@@endlink[0]{\pdfendlink}%
}%
\providecommand \url  [0]{\begingroup\@sanitize \@url }%
\providecommand \@url [1]{\endgroup\@href {#1}{\urlprefix}}%
\providecommand \urlprefix [0]{URL }%
\providecommand \Eprint[0]{\href }%
\@ifxundefined \urlstyle {%
  \providecommand \doi [1]{doi:\discretionary{}{}{}#1}%
}{%
  \providecommand \doi [0]{doi:\discretionary{}{}{}\begingroup
  \urlstyle{rm}\Url }%
}%
\providecommand \doibase [0]{http://dx.doi.org/}%
\providecommand \Doi[1]{\href{\doibase#1}}%
\providecommand \bibAnnote [3]{%
  \BibitemShut{#1}%
  \begin{quotation}\noindent
    \textsc{Key:}\ #2\\\textsc{Annotation:}\ #3%
  \end{quotation}%
}%
\providecommand \bibAnnoteFile [2]{%
  \IfFileExists{#2}{\bibAnnote {#1} {#2} {\input{#2}}}{}%
}%
\providecommand \typeout [0]{\immediate \write \m@ne }%
\providecommand \selectlanguage [0]{\@gobble}%
\providecommand \bibinfo [0]{\@secondoftwo}%
\providecommand \bibfield [0]{\@secondoftwo}%
\providecommand \translation [1]{[#1]}%
\providecommand \BibitemOpen[0]{}%
\providecommand \bibitemStop [0]{}%
\providecommand \bibitemNoStop [0]{.\EOS\space}%
\providecommand \EOS [0]{\spacefactor3000\relax}%
\providecommand \BibitemShut [1]{\csname bibitem#1\endcsname}%
\bibitem{Lazer:2009p10015}%
  \BibitemOpen
  \bibfield{author}{%
  \bibinfo {author} {\bibfnamefont{D.}~\bibnamefont{Lazer}}, \bibinfo {author}
  {\bibfnamefont{A.~S.}\ \bibnamefont{Pentland}}, \bibinfo {author}
  {\bibfnamefont{L.}~\bibnamefont{Adamic}}, \bibinfo {author}
  {\bibfnamefont{S.}~\bibnamefont{Aral}}, \bibinfo {author}
  {\bibfnamefont{A.~L.}\ \bibnamefont{Barabasi}}, \bibinfo {author}
  {\bibfnamefont{D.}~\bibnamefont{Brewer}}, \bibinfo {author}
  {\bibfnamefont{N.}~\bibnamefont{Christakis}}, \bibinfo {author}
  {\bibfnamefont{N.}~\bibnamefont{Contractor}}, \bibinfo {author}
  {\bibfnamefont{J.}~\bibnamefont{Fowler}}, \bibinfo {author}
  {\bibfnamefont{M.}~\bibnamefont{Gutmann}}, \bibinfo {author}
  {\bibfnamefont{T.}~\bibnamefont{Jebara}}, \bibinfo {author}
  {\bibfnamefont{G.}~\bibnamefont{King}}, \bibinfo {author}
  {\bibfnamefont{M.}~\bibnamefont{Macy}}, \bibinfo {author}
  {\bibfnamefont{D.}~\bibnamefont{Roy}},\ and\ \bibinfo {author}
  {\bibfnamefont{M.~V.}\ \bibnamefont{Alstyne}},\ }%
  \bibfield{journal}{%
  \Doi{10.1126/science.1167742}{\bibinfo {journal} {Science}}\ }%
  \textbf{\bibinfo {volume} {323}},\ \bibinfo {pages} {721} (\bibinfo {month}
  {Feb}\ \bibinfo {year} {2009})%
  \bibAnnoteFile{NoStop}{Lazer:2009p10015}%
\bibitem{Watts:2007p10016}%
  \BibitemOpen
  \bibfield{author}{%
  \bibinfo {author} {\bibfnamefont{D.~J.}\ \bibnamefont{Watts}},\ }%
  \bibfield{journal}{%
  \Doi{doi:10.1038/445489a}{\bibinfo {journal} {Nature}}\ }%
  \textbf{\bibinfo {volume} {445}},\ \bibinfo {pages} {489} (\bibinfo {month}
  {Jan}\ \bibinfo {year} {2007})%
  \bibAnnoteFile{NoStop}{Watts:2007p10016}%
\bibitem{Newman:2003p432}%
  \BibitemOpen
  \bibfield{author}{%
  \bibinfo {author} {\bibfnamefont{M.~E.~J.}\ \bibnamefont{Newman}},\ }%
  \bibfield{journal}{%
  \Doi{10.1137/S003614450342480}{\bibinfo {journal} {SIAM Rev}}\ }%
  \textbf{\bibinfo {volume} {45}},\ \bibinfo {pages} {167} (\bibinfo {year}
  {2003})%
  \bibAnnoteFile{NoStop}{Newman:2003p432}%
\bibitem{Fortunato:2010p6733}%
  \BibitemOpen
  \bibfield{author}{%
  \bibinfo {author} {\bibfnamefont{S.}~\bibnamefont{Fortunato}},\ }%
  \bibfield{journal}{%
  \Doi{10.1016/j.physrep.2009.11.002}{\bibinfo {journal} {Phys Rep}}\ }%
  \textbf{\bibinfo {volume} {486}},\ \bibinfo {pages} {75} (\bibinfo {month}
  {Feb}\ \bibinfo {year} {2010})%
  \bibAnnoteFile{NoStop}{Fortunato:2010p6733}%
\bibitem{Porter:2009p329}%
  \BibitemOpen
  \bibfield{author}{%
  \bibinfo {author} {\bibfnamefont{M.~A.}\ \bibnamefont{Porter}}, \bibinfo
  {author} {\bibfnamefont{J.-P.}\ \bibnamefont{Onnela}},\ and\ \bibinfo
  {author} {\bibfnamefont{P.~J.}\ \bibnamefont{Mucha}},\ }%
  \bibfield{journal}{%
  \bibinfo {journal} {Not Am Math Soc}\ }%
  \textbf{\bibinfo {volume} {56}},\ \bibinfo {pages} {1082} (\bibinfo {month}
  {Jan}\ \bibinfo {year} {2009})%
  \bibAnnoteFile{NoStop}{Porter:2009p329}%
\bibitem{Guimera2005Functional}%
  \BibitemOpen
  \bibfield{author}{%
  \bibinfo {author} {\bibfnamefont{R.}~\bibnamefont{Guimer{\`a}}}\ and\
  \bibinfo {author} {\bibfnamefont{L.~A.~N.}\ \bibnamefont{Amaral}},\ }%
  \bibfield{journal}{%
  \Doi{10.1038/nature03288}{\bibinfo {journal} {Nature}}\ }%
  \textbf{\bibinfo {volume} {433}},\ \bibinfo {pages} {895} (\bibinfo {month}
  {Feb}\ \bibinfo {year} {2005})%
  \bibAnnoteFile{NoStop}{Guimera2005Functional}%
\bibitem{Blondel2008Fast}%
  \BibitemOpen
  \bibfield{author}{%
  \bibinfo {author} {\bibfnamefont{V.~D.}\ \bibnamefont{Blondel}}, \bibinfo
  {author} {\bibfnamefont{J.-L.}\ \bibnamefont{Guillaume}}, \bibinfo {author}
  {\bibfnamefont{R.}~\bibnamefont{Lambiotte}},\ and\ \bibinfo {author}
  {\bibfnamefont{E.}~\bibnamefont{Lefebvre}},\ }%
  \bibfield{journal}{%
  \Doi{10.1088/1742-5468/2008/10/P10008}{\bibinfo {journal} {J Stat Mech-Theory
  E}}\ }%
  \textbf{\bibinfo {volume} {2008}},\ \bibinfo {pages} {P10008} (\bibinfo
  {month} {Oct}\ \bibinfo {year} {2008})%
  \bibAnnoteFile{NoStop}{Blondel2008Fast}%
\bibitem{Guimera2005Worldwide}%
  \BibitemOpen
  \bibfield{author}{%
  \bibinfo {author} {\bibfnamefont{R.}~\bibnamefont{Guimer{\`a}}}, \bibinfo
  {author} {\bibfnamefont{S.}~\bibnamefont{Mossa}}, \bibinfo {author}
  {\bibfnamefont{A.}~\bibnamefont{Turtschi}},\ and\ \bibinfo {author}
  {\bibfnamefont{L.}~\bibnamefont{Amaral}},\ }%
  \bibfield{journal}{%
  \Doi{10.1073/pnas.0407994102}{\bibinfo {journal} {Proc Natl Acad Sci USA}}\
  }%
  \textbf{\bibinfo {volume} {102}},\ \bibinfo {pages} {7794} (\bibinfo {month}
  {May}\ \bibinfo {year} {2005})%
  \bibAnnoteFile{NoStop}{Guimera2005Worldwide}%
\bibitem{Newman2004Finding}%
  \BibitemOpen
  \bibfield{author}{%
  \bibinfo {author} {\bibfnamefont{M.~E.~J.}\ \bibnamefont{Newman}}\ and\
  \bibinfo {author} {\bibfnamefont{M.}~\bibnamefont{Girvan}},\ }%
  \bibfield{journal}{%
  \Doi{10.1103/PhysRevE.69.026113}{\bibinfo {journal} {Phys Rev E}}\ }%
  \textbf{\bibinfo {volume} {69}},\ \bibinfo {pages} {026113} (\bibinfo {month}
  {Feb}\ \bibinfo {year} {2004})%
  \bibAnnoteFile{NoStop}{Newman2004Finding}%
\bibitem{Barber:2007p7162}%
  \BibitemOpen
  \bibfield{author}{%
  \bibinfo {author} {\bibfnamefont{M.~J.}\ \bibnamefont{Barber}},\ }%
  \bibfield{journal}{%
  \Doi{10.1103/PhysRevE.76.066102}{\bibinfo {journal} {Phys Rev E}}\ }%
  \textbf{\bibinfo {volume} {76}},\ \bibinfo {pages} {066102} (\bibinfo {month}
  {Dec}\ \bibinfo {year} {2007})%
  \bibAnnoteFile{NoStop}{Barber:2007p7162}%
\bibitem{Traag:2009p147}%
  \BibitemOpen
  \bibfield{author}{%
  \bibinfo {author} {\bibfnamefont{V.~A.}\ \bibnamefont{Traag}}\ and\ \bibinfo
  {author} {\bibfnamefont{J.}~\bibnamefont{Bruggeman}},\ }%
  \bibfield{journal}{%
  \Doi{10.1103/PhysRevE.80.036115}{\bibinfo {journal} {Phys Rev E}}\ }%
  \textbf{\bibinfo {volume} {80}},\ \bibinfo {pages} {036115} (\bibinfo {month}
  {Sep}\ \bibinfo {year} {2009})%
  \bibAnnoteFile{NoStop}{Traag:2009p147}%
\bibitem{Mucha:2010p321}%
  \BibitemOpen
  \bibfield{author}{%
  \bibinfo {author} {\bibfnamefont{P.~J.}\ \bibnamefont{Mucha}}, \bibinfo
  {author} {\bibfnamefont{T.}~\bibnamefont{Richardson}}, \bibinfo {author}
  {\bibfnamefont{K.}~\bibnamefont{Macon}}, \bibinfo {author}
  {\bibfnamefont{M.~A.}\ \bibnamefont{Porter}},\ and\ \bibinfo {author}
  {\bibfnamefont{J.-P.}\ \bibnamefont{Onnela}},\ }%
  \bibfield{journal}{%
  \Doi{10.1126/science.1184819}{\bibinfo {journal} {Science}}\ }%
  \textbf{\bibinfo {volume} {328}},\ \bibinfo {pages} {876} (\bibinfo {month}
  {May}\ \bibinfo {year} {2010})%
  \bibAnnoteFile{NoStop}{Mucha:2010p321}%
\bibitem{Lancichinetti:2008p7072}%
  \BibitemOpen
  \bibfield{author}{%
  \bibinfo {author} {\bibfnamefont{A.}~\bibnamefont{Lancichinetti}}, \bibinfo
  {author} {\bibfnamefont{S.}~\bibnamefont{Fortunato}},\ and\ \bibinfo {author}
  {\bibfnamefont{F.}~\bibnamefont{Radicchi}},\ }%
  \bibfield{journal}{%
  \Doi{10.1103/PhysRevE.78.046110}{\bibinfo {journal} {Phys Rev E}}\ }%
  \textbf{\bibinfo {volume} {78}},\ \bibinfo {pages} {046110} (\bibinfo {month}
  {Oct}\ \bibinfo {year} {2008})%
  \bibAnnoteFile{NoStop}{Lancichinetti:2008p7072}%
\bibitem{Fortunato:2007p183}%
  \BibitemOpen
  \bibfield{author}{%
  \bibinfo {author} {\bibfnamefont{S.}~\bibnamefont{Fortunato}}\ and\ \bibinfo
  {author} {\bibfnamefont{M.}~\bibnamefont{Barthelemy}},\ }%
  \bibfield{journal}{%
  \Doi{10.1073/pnas.0605965104}{\bibinfo {journal} {Proc Natl Acad Sci USA}}\
  }%
  \textbf{\bibinfo {volume} {104}},\ \bibinfo {pages} {36} (\bibinfo {month}
  {Jan}\ \bibinfo {year} {2007})%
  \bibAnnoteFile{NoStop}{Fortunato:2007p183}%
\bibitem{Good:2010p10071}%
  \BibitemOpen
  \bibfield{author}{%
  \bibinfo {author} {\bibfnamefont{B.~H.}\ \bibnamefont{Good}}, \bibinfo
  {author} {\bibfnamefont{Y.-A.}\ \bibnamefont{de~Montjoye}},\ and\ \bibinfo
  {author} {\bibfnamefont{A.}~\bibnamefont{Clauset}},\ }%
  \bibfield{journal}{%
  \Doi{10.1103/PhysRevE.81.046106}{\bibinfo {journal} {Phys Rev E}}\ }%
  \textbf{\bibinfo {volume} {81}},\ \bibinfo {pages} {046106} (\bibinfo {month}
  {Apr}\ \bibinfo {year} {2010})%
  \bibAnnoteFile{NoStop}{Good:2010p10071}%
\bibitem{kumpala07}%
  \BibitemOpen
  \bibfield{author}{%
  \bibinfo {author} {\bibfnamefont{J.~M.}\ \bibnamefont{Kumpula}}, \bibinfo
  {author} {\bibfnamefont{J.}~\bibnamefont{Saram{\"a}ki}}, \bibinfo {author}
  {\bibfnamefont{K.}~\bibnamefont{Kaski}},\ and\ \bibinfo {author}
  {\bibfnamefont{J.}~\bibnamefont{Kert{\'e}sz}},\ }%
  \bibfield{journal}{%
  \Doi{10.1140/epjb/e2007-00088-4}{\bibinfo {journal} {Eur Phys J B}}\ }%
  \textbf{\bibinfo {volume} {56}},\ \bibinfo {pages} {41} (\bibinfo {month}
  {Mar}\ \bibinfo {year} {2007})%
  \bibAnnoteFile{NoStop}{kumpala07}%
\bibitem{Delvenne:2010p11939}%
  \BibitemOpen
  \bibfield{author}{%
  \bibinfo {author} {\bibfnamefont{J.-C.}\ \bibnamefont{Delvenne}}, \bibinfo
  {author} {\bibfnamefont{S.}~\bibnamefont{Yaliraki}},\ and\ \bibinfo {author}
  {\bibfnamefont{M.}~\bibnamefont{Barahona}},\ }%
  \bibfield{journal}{%
  \Doi{10.1073/pnas.0903215107}{\bibinfo {journal} {Proc Natl Acad Sci USA}}\
  }%
  \textbf{\bibinfo {volume} {107}},\ \bibinfo {pages} {12755} (\bibinfo {month}
  {Jul}\ \bibinfo {year} {2010})%
  \bibAnnoteFile{NoStop}{Delvenne:2010p11939}%
\bibitem{Arenas:2008p12974}%
  \BibitemOpen
  \bibfield{author}{%
  \bibinfo {author} {\bibfnamefont{A.}~\bibnamefont{Arenas}}, \bibinfo {author}
  {\bibfnamefont{A.}~\bibnamefont{Fern{\'a}ndez}},\ and\ \bibinfo {author}
  {\bibfnamefont{S.}~\bibnamefont{G{\'o}mez}},\ }%
  \bibfield{journal}{%
  \Doi{10.1088/1367-2630/10/5/053039}{\bibinfo {journal} {New J Phys}}\ }%
  \textbf{\bibinfo {volume} {10}},\ \bibinfo {pages} {053039} (\bibinfo {month}
  {May}\ \bibinfo {year} {2008})%
  \bibAnnoteFile{NoStop}{Arenas:2008p12974}%
\bibitem{Ronhovde:2010p7486}%
  \BibitemOpen
  \bibfield{author}{%
  \bibinfo {author} {\bibfnamefont{P.}~\bibnamefont{Ronhovde}}\ and\ \bibinfo
  {author} {\bibfnamefont{Z.}~\bibnamefont{Nussinov}},\ }%
  \bibfield{journal}{%
  \Doi{doi:10.1103/PhysRevE.81.046114}{\bibinfo {journal} {Phys Rev E}}\ }%
  \textbf{\bibinfo {volume} {81}},\ \bibinfo {pages} {046114} (\bibinfo {month}
  {Jan}\ \bibinfo {year} {2010})%
  \bibAnnoteFile{NoStop}{Ronhovde:2010p7486}%
\bibitem{Reichardt2007Partitioning}%
  \BibitemOpen
  \bibfield{author}{%
  \bibinfo {author} {\bibfnamefont{J.}~\bibnamefont{Reichardt}}\ and\ \bibinfo
  {author} {\bibfnamefont{S.}~\bibnamefont{Bornholdt}},\ }%
  \bibfield{journal}{%
  \Doi{10.1103/PhysRevE.76.015102}{\bibinfo {journal} {Phys Rev E}}\ }%
  \textbf{\bibinfo {volume} {76}},\ \bibinfo {pages} {015102+} (\bibinfo {year}
  {2007})%
  \bibAnnoteFile{NoStop}{Reichardt2007Partitioning}%
\bibitem{Raghavan:2007p11294}%
  \BibitemOpen
  \bibfield{author}{%
  \bibinfo {author} {\bibfnamefont{U.~N.}\ \bibnamefont{Raghavan}}, \bibinfo
  {author} {\bibfnamefont{R.}~\bibnamefont{Albert}},\ and\ \bibinfo {author}
  {\bibfnamefont{S.}~\bibnamefont{Kumara}},\ }%
  \bibfield{journal}{%
  \Doi{10.1103/PhysRevE.76.036106}{\bibinfo {journal} {Phys Rev E}}\ }%
  \textbf{\bibinfo {volume} {76}},\ \bibinfo {pages} {036106} (\bibinfo {month}
  {Sep}\ \bibinfo {year} {2007})%
  \bibAnnoteFile{NoStop}{Raghavan:2007p11294}%
\bibitem{Tibely:2008p10903}%
  \BibitemOpen
  \bibfield{author}{%
  \bibinfo {author} {\bibfnamefont{G.}~\bibnamefont{Tibely}}\ and\ \bibinfo
  {author} {\bibfnamefont{J.}~\bibnamefont{Kert{\'e}sz}},\ }%
  \bibfield{journal}{%
  \Doi{10.1016/j.physa.2008.04.024}{\bibinfo {journal} {Physica A}}\ }%
  \textbf{\bibinfo {volume} {387}},\ \bibinfo {pages} {4982} (\bibinfo {month}
  {Aug}\ \bibinfo {year} {2008})%
  \bibAnnoteFile{NoStop}{Tibely:2008p10903}%
\bibitem{Reichardt:2004p11243}%
  \BibitemOpen
  \bibfield{author}{%
  \bibinfo {author} {\bibfnamefont{J.}~\bibnamefont{Reichardt}}\ and\ \bibinfo
  {author} {\bibfnamefont{S.}~\bibnamefont{Bornholdt}},\ }%
  \bibfield{journal}{%
  \Doi{10.1103/PhysRevLett.93.218701}{\bibinfo {journal} {Phys Rev Lett}}\ }%
  \textbf{\bibinfo {volume} {93}},\ \bibinfo {pages} {218701} (\bibinfo {month}
  {Nov}\ \bibinfo {year} {2004})%
  \bibAnnoteFile{NoStop}{Reichardt:2004p11243}%
\bibitem{Brandes:2007p196}%
  \BibitemOpen
  \bibfield{author}{%
  \bibinfo {author} {\bibfnamefont{U.}~\bibnamefont{Brandes}}, \bibinfo
  {author} {\bibfnamefont{D.}~\bibnamefont{Delling}}, \bibinfo {author}
  {\bibfnamefont{M.}~\bibnamefont{Gaertler}}, \bibinfo {author}
  {\bibfnamefont{R.}~\bibnamefont{G{\"o}rke}}, \bibinfo {author}
  {\bibfnamefont{M.}~\bibnamefont{Hoefer}}, \bibinfo {author}
  {\bibfnamefont{Z.}~\bibnamefont{Nikoloski}},\ and\ \bibinfo {author}
  {\bibfnamefont{D.}~\bibnamefont{Wagner}},\ }%
  \bibfield{journal}{%
  \Doi{10.1007/978-3-540-74839-7_12}{\bibinfo {journal} {Lect Notes Comput
  Sc}}\ }%
  \textbf{\bibinfo {volume} {4769}},\ \bibinfo {pages} {121} (\bibinfo {month}
  {Jan}\ \bibinfo {year} {2007})%
  \bibAnnoteFile{NoStop}{Brandes:2007p196}%
\bibitem{Krings:2011p12227}%
  \BibitemOpen
  \bibfield{author}{%
  \bibinfo {author} {\bibfnamefont{G.}~\bibnamefont{Krings}}\ and\ \bibinfo
  {author} {\bibfnamefont{V.~D.}\ \bibnamefont{Blondel}},\ }%
  \bibfield{journal}{%
  \bibinfo {journal} {Arxiv preprint arXiv:1103.5569}}%
   (\bibinfo {month} {Jan}\ \bibinfo {year} {2011})%
  \bibAnnoteFile{NoStop}{Krings:2011p12227}%
\bibitem{Expert:2011p12973}%
  \BibitemOpen
  \bibfield{author}{%
  \bibinfo {author} {\bibfnamefont{P.}~\bibnamefont{Expert}}, \bibinfo {author}
  {\bibfnamefont{T.}~\bibnamefont{Evans}}, \bibinfo {author}
  {\bibfnamefont{V.~D.}\ \bibnamefont{Blondel}},\ and\ \bibinfo {author}
  {\bibfnamefont{R.}~\bibnamefont{Lambiotte}},\ }%
  \bibfield{journal}{%
  \Doi{10.1073/pnas.1018962108}{\bibinfo {journal} {Proc Natl Acad Sci USA}}\
  }%
  \textbf{\bibinfo {volume} {108}},\ \bibinfo {pages} {7663} (\bibinfo {month}
  {Jan}\ \bibinfo {year} {2011})%
  \bibAnnoteFile{NoStop}{Expert:2011p12973}%
\bibitem{Lambiotte:2008p2987}%
  \BibitemOpen
  \bibfield{author}{%
  \bibinfo {author} {\bibfnamefont{R.}~\bibnamefont{Lambiotte}}, \bibinfo
  {author} {\bibfnamefont{V.}~\bibnamefont{Blondel}}, \bibinfo {author}
  {\bibfnamefont{C.}~\bibnamefont{Dekerchove}}, \bibinfo {author}
  {\bibfnamefont{E.}~\bibnamefont{Huens}}, \bibinfo {author}
  {\bibfnamefont{C.}~\bibnamefont{Prieur}}, \bibinfo {author}
  {\bibfnamefont{Z.}~\bibnamefont{Smoreda}},\ and\ \bibinfo {author}
  {\bibfnamefont{P.}~\bibnamefont{van Dooren}},\ }%
  \bibfield{journal}{%
  \Doi{10.1016/j.physa.2008.05.014}{\bibinfo {journal} {Physica A}}\ }%
  \textbf{\bibinfo {volume} {387}},\ \bibinfo {pages} {5317} (\bibinfo {month}
  {Sep}\ \bibinfo {year} {2008})%
  \bibAnnoteFile{NoStop}{Lambiotte:2008p2987}%
\bibitem{Lancichinetti:2009p7112}%
  \BibitemOpen
  \bibfield{author}{%
  \bibinfo {author} {\bibfnamefont{A.}~\bibnamefont{Lancichinetti}}\ and\
  \bibinfo {author} {\bibfnamefont{S.}~\bibnamefont{Fortunato}},\ }%
  \bibfield{journal}{%
  \Doi{10.1103/PhysRevE.80.056117}{\bibinfo {journal} {Phys Rev E}}\ }%
  \textbf{\bibinfo {volume} {80}},\ \bibinfo {pages} {056117} (\bibinfo {month}
  {Nov}\ \bibinfo {year} {2009})%
  \bibAnnoteFile{NoStop}{Lancichinetti:2009p7112}%
\bibitem{Rosvall2007Informationtheoretic}%
  \BibitemOpen
  \bibfield{author}{%
  \bibinfo {author} {\bibfnamefont{M.}~\bibnamefont{Rosvall}}\ and\ \bibinfo
  {author} {\bibfnamefont{C.~T.}\ \bibnamefont{Bergstrom}},\ }%
  \bibfield{journal}{%
  \Doi{10.1073/pnas.0611034104}{\bibinfo {journal} {Proc Natl Acad Sci USA}}\
  }%
  \textbf{\bibinfo {volume} {104}},\ \bibinfo {pages} {7327} (\bibinfo {month}
  {May}\ \bibinfo {year} {2007})%
  \bibAnnoteFile{NoStop}{Rosvall2007Informationtheoretic}%
\bibitem{Leicht2008Community}%
  \BibitemOpen
  \bibfield{author}{%
  \bibinfo {author} {\bibfnamefont{E.~A.}\ \bibnamefont{Leicht}}\ and\ \bibinfo
  {author} {\bibfnamefont{M.~E.~J.}\ \bibnamefont{Newman}},\ }%
  \bibfield{journal}{%
  \Doi{10.1103/PhysRevLett.100.118703}{\bibinfo {journal} {Phys Rev Lett}}\ }%
  \textbf{\bibinfo {volume} {100}},\ \bibinfo {pages} {118703+} (\bibinfo
  {year} {2008})%
  \bibAnnoteFile{NoStop}{Leicht2008Community}%
\bibitem{Lambiotte:2010p6743}%
  \BibitemOpen
  \bibfield{author}{%
  \bibinfo {author} {\bibfnamefont{R.}~\bibnamefont{Lambiotte}},\ }%
  \bibfield{journal}{%
  \bibinfo {journal} {Arxiv preprint arXiv:1004.4268}}%
   (\bibinfo {month} {Jan}\ \bibinfo {year} {2010})%
  \bibAnnoteFile{NoStop}{Lambiotte:2010p6743}%
\bibitem{Gfeller:2005p13024}%
  \BibitemOpen
  \bibfield{author}{%
  \bibinfo {author} {\bibfnamefont{D.}~\bibnamefont{Gfeller}}, \bibinfo
  {author} {\bibfnamefont{J.-C.}\ \bibnamefont{Chappelier}},\ and\ \bibinfo
  {author} {\bibfnamefont{P.}~\bibnamefont{de~Los~Rios}},\ }%
  \bibfield{journal}{%
  \Doi{10.1103/PhysRevE.72.056135}{\bibinfo {journal} {Phys Rev E}}\ }%
  \textbf{\bibinfo {volume} {72}},\ \bibinfo {pages} {056135} (\bibinfo {month}
  {Nov}\ \bibinfo {year} {2005})%
  \bibAnnoteFile{NoStop}{Gfeller:2005p13024}%
\bibitem{Bianconi:2009p13095}%
  \BibitemOpen
  \bibfield{author}{%
  \bibinfo {author} {\bibfnamefont{G.}~\bibnamefont{Bianconi}}, \bibinfo
  {author} {\bibfnamefont{P.}~\bibnamefont{Pin}},\ and\ \bibinfo {author}
  {\bibfnamefont{M.}~\bibnamefont{Marsili}},\ }%
  \bibfield{journal}{%
  \Doi{10.1073/pnas.0811511106}{\bibinfo {journal} {Proc Natl Acad Sci USA}}\
  }%
  \textbf{\bibinfo {volume} {106}},\ \bibinfo {pages} {11433} (\bibinfo {month}
  {Jul}\ \bibinfo {year} {2009})%
  \bibAnnoteFile{NoStop}{Bianconi:2009p13095}%
\bibitem{Rosvall:2011p11938}%
  \BibitemOpen
  \bibfield{author}{%
  \bibinfo {author} {\bibfnamefont{M.}~\bibnamefont{Rosvall}}\ and\ \bibinfo
  {author} {\bibfnamefont{C.~T.}\ \bibnamefont{Bergstrom}},\ }%
  \bibfield{journal}{%
  \Doi{10.1371/journal.pone.0018209}{\bibinfo {journal} {PloS one}}\ }%
  \textbf{\bibinfo {volume} {6}},\ \bibinfo {pages} {e18209} (\bibinfo {month}
  {Jan}\ \bibinfo {year} {2011})%
  \bibAnnoteFile{NoStop}{Rosvall:2011p11938}%
\end{thebibliography}
\end{document}